 \documentclass[11pt]{article}

\usepackage{authblk}
\usepackage{amsmath,amssymb}
\usepackage{epsfig, graphicx,fleqn,xspace,color,graphics}

\usepackage{helvet,graphicx,makeidx,multicol}

\usepackage{float}
\usepackage{setspace}
\usepackage{cite}

\usepackage{xcolor}

\usepackage{algorithm}
\usepackage{framed}
\usepackage{algpseudocode}
\newfloat{alg}{thp}{lop}
\floatname{alg}{Algorithm}

\usepackage{xspace}
\newenvironment{proof}{\noindent{\bf Proof.} }{\null\hfill$\Box$\par\medskip}

\newtheorem{theorem}{Theorem}[section]

\newtheorem{lemma}{Lemma}[section]
\newtheorem{corollary}{Corollary}[section]

\newtheorem{conjecture}{Conjecture}[section]

\newtheorem{definition}{Definition}[section]

\newcounter{smallitemizec}

\newfloat{code}{thp}{lop}
\date{}
\begin{document}
\title{Exploration of High-Dimensional Grids by Finite State Machines}
\author[1]{Stefan Dobrev}
\author[2]{Lata Narayanan} 
\author[2]{Jaroslav Opatrny} 
\author[2]{Denis Pankratov}
\affil[1]{Institute of Mathematics, Slovak Academy of Sciences, Bratislava, Slovakia}
\affil[2]{Depertment of CSSE, Concordia University, Montreal, Canada}
\maketitle

\begin{abstract}
We consider the problem of finding a treasure at an unknown point of an $n$-dimensional infinite grid, $n\geq 3$,  by initially collocated finite state agents (scouts/robots).
Recently, the problem has been well characterized for 2 dimensions for deterministic as well as randomized agents, both in synchronous and semi-synchronous models~\cite{Brandt2018,EmekLSUW15}. 
It has been conjectured that $n+1$ randomized agents are necessary to solve this problem in the $n$-dimensional grid~\cite{CohenELU2017}. 
In this paper we disprove the conjecture in a strong sense: we show that {\em three} randomized synchronous agents suffice to explore an $n$-dimensional grid for {\em any} $n$.
Our algorithm is optimal in terms of the number of the agents. 
Our key insight is that a constant number of finite state machine agents can, by their positions and movements, implement a stack, which can store the path being explored. We also show how to implement our algorithm using:  four randomized semi-synchronous agents; four deterministic synchronous agents; or five deterministic semi-synchronous agents. 

We give a different algorithm that uses $4$ deterministic semi-synchronous agents  for the $3$-dimensional grid. This is provably optimal, and surprisingly,  matches the result for $2$ dimensions. 
%
 For $n\geq 4$, the time complexity of the solutions mentioned above is exponential in distance  $D$ of the treasure from the starting point of the agents. We show that in the deterministic case, one additional agent brings the time down to a polynomial.
Finally, we focus on algorithms that never venture much beyond the distance $D$. We describe an algorithm that uses $O(\sqrt{n})$ semi-synchronous deterministic agents that never go beyond $2D$, as well as show that any algorithm using $3$ synchronous deterministic agents in $3$ dimensions must travel beyond $\Omega(D^{3/2})$ from the origin.
\end{abstract}

\section{Introduction} 
\label{sec:intro}

 Motivated by the self-organizing behaviour of ants and other social insects, swarm robotics leverages  the collective capability of a collection of extremely simple and inexpensive robots. Such robots have very limited computation and
communication capabilities, and yet can collectively perform seemingly complex tasks such as: forage for food \cite{CYLLWN2016}; form patterns \cite{FPSW2008};    pull heavy objects \cite{EMCCF18}; and  play {\em F\"{u}r Elise} on the piano \cite{CE2012}.  

A series of recent papers \cite{FeinermanKLS13,EmekLUW14,EmekLSUW15,Brandt2018,CohenELU2017} studies the conditions required for such primitive robots (also called agents or scouts) to search for a treasure placed at an unknown location in an infinite two-dimensional grid.  In particular,  they consider agents whose
behaviour is controlled by a finite automaton (FA), and that can only communicate with other agents that are at the exact same grid location as themselves. Furthermore, this communication is limited to a constant number of bits. The primary question of interest is: {\em how many } such agents are needed to search for a treasure located at an unknown location in an infinite $n$-dimensional grid for $n \geq 2$?
As shown in \cite{Brandt2018,EmekLSUW15} for $n=2$, the answer depends on the computational power of the agents: whether or not they have  access to random bits, the amount of memory
they have, and whether or not they are synchronized.  Note that for randomized algorithms, we require a finite mean hitting time for every node in the grid. The set of agents is {\em fully synchronous} if they operate by the same global clock; they are {\em semi-synchronous}\footnote{In some related literature~\cite{EmekLSUW15,CohenELU2017,EmekLSUW15} the same model was referred to as asynchronous. We follow the terminology of semi-synchronous of \cite{Brandt2018} and the vast literature on autonomous mobile robots to avoid confusion with a fully asynchronous model.} if in every time slot, a subset of adversarially scheduled agents is active. In our algorithms, all agents are finite automata, full details of the agent models are given in Section~\ref{sec:model}.

The case of the 2-dimensional grid has been completely characterized. It has been shown that  if the agents are deterministic and semi-synchronous, 4 agents are necessary \cite{Brandt2018} and sufficient \cite{EmekLSUW15}. If random bits are available to the agents,    3 agents are necessary \cite{CohenELU2017} and sufficient \cite{EmekLSUW15}, regardless of whether they are synchronous or semi-synchronous. Even without random bits, if the agents are fully synchronous, then 3 agents are necessary and sufficient \cite{EmekLSUW15}.

In \cite{CohenELU2017}, the authors proved that 3 agents are necessary to search the 2-dimensional grid, even if they are fully synchronized and are randomized. They conjectured that in an $n$-dimensional grid, $n+1$ agents would be necessary. 

\begin{conjecture} \cite{CohenELU2017}
For $n \geq 3$, any search strategy on the $n$-dimensional infinite grid requires at least $n+1$ agents. 
\end{conjecture}

The main result of this paper is to disprove the above conjecture; we show that three randomized synchronous agents, or 5 deterministic semi-synchronous agents can explore any $n$-dimensional grid. These algorithms  are completely different from previous algorithms for grid exploration, and are based on the key insight that a constant number of finite state machine agents can, by their positions and movements, implement a stack that stores the path being explored. 

\subsection{Our results}

First, we show that 
 in the $3$-dimensional grid, $4$ deterministic semi-synchronous agents are sufficient for grid exploration. We give  an algorithm which is similar to  the algorithm for the $2$-dimensional grid given in \cite{EmekLSUW15}, but with an important modification that enables exploration of the 3-dimensional grid without increasing the number of agents. Our algorithm is optimal in the explored space and also in the number of agents, since 4 agents are necessary to explore even the 2-dimensional grid.

Our main result is an algorithm for 3 randomized synchronous agents  to explore an $n$-dimensional grid for any $n \geq 3$. This result is optimal, since 3 agents are necessary to explore even the 2-dimensional grid. Next we show how to "derandomize" the algorithm with the addition of one agent. If the agents are semi-synchronous, the algorithm can be implemented with the addition of one more agent, in both the randomized and deterministic cases. Table~\ref{table:summary} shows our results.

\begin{table}[h] \label{table:summary}
\begin{center}
\begin{tabular}{||c|c|l||}
\hline
\hline
Model & Number of agents & Section \\ \hline\hline
Randomized Synchronous & $3^*$ & Section~\ref{SSec:nRandom} \\ \hline
Randomized Semi-synchronous & 4 & Section~\ref{SSec:nRandom} \\ \hline
Deterministic Synchronous & 4 & Section~\ref{SSec:det} \\ \hline
Deterministic Semi-synchronous & \begin{tabular}{@{}l@{}}$4^*$ ($n=3$)  \\ 5  ($n \geq 3)$ \end{tabular} & \begin{tabular}{@{}l@{}} Section~\ref{sec:3dim} \\ Section~\ref{SSec:det}  \end{tabular} \\ \hline \hline
\end{tabular}
\caption{Exploration of an $n$ dimensional infinite grid. Numbers marked with $^*$ indicate that the optimal number of agents is used.}
\end{center}
\end{table}

The algorithms mentioned in Table~\ref{table:summary}, except the 4-agent deterministic semi-synchronous algorithm for the three-dimensional grid, have an exploration cost/time  that is exponential in the volume of the smallest ball containing the treasure. In Section~\ref{sec:ndimpol}, we give a deterministic
synchronous algorithm for exploring the $n$-dimensional grid that uses 5 agents and  takes time polynomial in $D$, the distance from the origin to the treasure. A semi-synchronous implementation of this algorithm uses 6 agents. In Section~\ref{sec:bounded}, we give a lower bound of $\Omega(D^{3/2})$ on the distance from the origin that must be travelled by some agent in any 3-agent deterministic synchronous algorithm, and give an algorithm using $O(\sqrt{n})$  deterministic semi-synchronous agents in which no agent travels distance more than $2D$.

\subsection{Related work}

First introduced by Beck~\cite{Beck1964} and Bellman~\cite{bellman1963optimal},  the cow-path problem is the problem of minimizing the time required for search for a treasure on an infinite line by a single agent. Since then many variants have been studied, including search on the plane, and by multiple robots \cite{Baeza-YatesCR88,Baeza-YatesS95,Beck1965more,Beck1973return,BrandtFRW17, ChrobakGGM15,GhoshK10,JezL09,LiC09a,Lopez-OrtizS01,Bose13}. 
Evacuation, or group search by a set of collaborating robots where the objective is to minimize the time
the {\em last} robot arrives at the treasure has been the focus of many recent papers (see for example \cite{ChrobakGGM15,CzyzowiczGGKMP14}). Two models of communication have been studied for the collaborating robots: wireless, or face-to-face. In the latter model, similar to the model we use in this paper, the robots communicate only if they are at the same place at the same time. 

Graph exploration is a much-studied problem (see, for example, \cite{Deng90,Fraigniaud2005,BENDER20021,Bender1994,Brass2011,flocchini2010} and references therein). The study of the exploration of labyrinths  by agents with pebbles is related to our work. A labyrinth is a two-dimensional grid with some blocked cells; it is called finite if a finite number of cells is blocked. It was shown in \cite{Blum77} that finite 2D labyrinths can be explored by one FA agent with four pebbles, but no collection of FA agents can search 3D maze.
Our algorithm in Section~\ref{sec:ndim} 
can be implemented using a single FA agent with four pebbles to explore $n$-dimensional grids, albeit with no blocked nodes. Recently it was shown that $\Theta(\log \log n)$ pebbles are necessary and sufficient for exploration of arbitrary unknown undirected graphs by a single FA agent. 

Aleliunas {\em et al} \cite{Al79} showed that a  random walk by a single agent has a polynomial hitting time on a finite graph. On the infinite $n$-dimensional grid, it is known that every node on the grid can be reached with probability 1 if and only if $n \leq 2$.  However, the mean hitting time of some nodes is infinite, even for $n=1$. Exploration with 3 non-interacting random walks achieves a finite mean hitting time for all nodes a one-dimensional grid, but on a two-dimensional grid, this is not possible with any finite number of non-interacting random walks. 

A large body of work is devoted to the capabilities of autonomous mobile robots with very limited computational and communication abilities; see \cite{FPS12} for a comprehensive introduction. While we borrow some of the terminology in Section~\ref{sec:model}, their robots are usually assumed to be identical, anonymous, and communication is limited to being able to "see" each other's positions, regardless of how far they are. In contrast, in our model, the robots follow different algorithms (or they can be assumed to start at different states of the same FSM), can only communicate if they are at the same location, though the communication is limited to a constant number of bits. Equivalently, they can be assumed to see the current states of other robots at the same location. This is similar to the "robots with lights" model in the autonomous mobile robot literature \cite{DasFPSY2016}. 

The research most related to our work was initiated by Feinerman {\em et al}  in \cite{FeinermanKLS13}, which introduced the problem of $k$ randomized mobile agents, starting from the same initial position, and searching for a treasure at an unknown location on the two-dimensional infinite grid. In their model, the agents are Turing machines, but cannot communicate at all. They show that if the agents have a constant approximation of $k$, the treasure can be found optimally in time $O(D + D^2/k)$, where $D$ is the distance between the initial location and the treasure. The authors of  \cite{EmekLUW14} consider semi-synchronous and randomized FA agents and show that the same time complexity can be achieved. The relationship between the number of random bits available and the search time was studied in \cite{LenzenLNR14}. 

Emek  {\em et al} \cite{EmekLSUW15} posed the question of {\em how many} agents are required to find the treasure. They studied deterministic as well as randomized agents, synchronous as well as semi-synchronous agents, and FA agents, as well as agents that are controlled by a push-down automaton (PDA). They show that the problem can be solved by any of the following: 4 deterministic semi-synchronous FA agents; 3 deterministic synchronous agents;  3 randomized semi-synchronous
FA agents; 1 deterministic FA together with 1 deterministic PDA agent; 1 randomized PDA agent. On the negative side they show that the problem cannot be solved by 2 deterministic (synchronous) FA agents;  a single randomized FA agent; a deterministic PDA agent. Cohen {\em et al} \cite{CohenELU2017} prove that at least 2 FA agents are necessary to explore the one-dimensional grid and at least 3 FA agents are needed to explore the two-dimensional grid, thus proving the optimality of the FA-agent deterministic synchronous and randomized semi-synchronous
algorithms in \cite{EmekLSUW15}. Recently it was shown that 3 deterministic semi-synchronous FA  agents cannot perform exploration of the 2-dimensional grid \cite{Brandt2018}, thus proving the optimality of the 4 FA-agent deterministic synchronous algorithm in \cite{EmekLSUW15}.

\section{Model and Notation}
\label{sec:model}
We use the same models (with an exception of Section~\ref{sec:unoriented} on unoriented grids) as in \cite{EmekLSUW15,Brandt2018}. For completeness, we recall key definitions and introduce some notation in this section.

Our search domain is $\mathbb{Z}^n$ with the \emph{Manhattan metric}, i.e., the distance between two points $p, p' \in \mathbb{Z}^n$ is defined as $||p - p'|| = \sum_{i=1}^n |p_i - p'_i|$. We refer to $\mathbb{Z}^n$ as the \emph{$n$-dimensional integer grid} and its elements as \emph{grid points, points, or cells}. A grid point $p=(p_1,p_2, \ldots p_i,\ldots, p_n)$ is \emph{adjacent} to every grid point $(p_1,p_2, \ldots p'_i,\ldots, p_n)$, where $|p_i-p_i'|=1$ for some $i$, $1\leq i\leq n$. Thus $p$ is adjacent to grid points whose coordinates differ from those of $p$ in exactly one dimension and exactly by 1. We assume that any two grid points cannot be distinguished from each other by an agent, and that includes the origin from which the search starts.

The search for the treasure in the grid is done using a fixed number of agents. We assume that each agent is a very simple device of very limited communication capabilities. Thus, an agent is modelled by a finite automaton, and two agents can exchange information with each other only when they occupy the same grid location at the same time. Initially, all agents are located in the same grid point. Without loss of generality we assume that this cell is the origin of the grid. The treasure is located at distance $D$ from from the origin and this distance is not known to the agents. We assume that 
the grid is oriented and the edges out of each grid point are labelled by dimensions.

Time is divided into discrete units. In each time unit an {\em active} agent performs a single {\em look-compute-move} cycle. In the look part of the cycle the agent sees the state of other agents located in its own grid point. At the compute part of the cycle the agent determines, using its own state and those it sees, to which adjacent node to move to, if at all. The agent also determines its new state. Such a move is then executed in the move part of the cycle. When we consider randomized algorithms, we assume that an agent has access to a random value during each compute cycle, as needed. 

We say that the system is {\em synchronous} if at each time unit all agents are active. We say that the system is {\em semi-synchronous} if at each time unit only a subset of agents, chosen by an adversarial scheduler, is active. In order to avoid trivial cases, one restriction on the adversarial scheduler is introduced --- it must schedule each agent infinitely often.

In addition to the question of whether $\mathbb{Z}^n$ can be fully explored by $k$ agents, we are also interested in the efficiency of such exploration procedures.  We refer to this measure of interest as \emph{the exploration cost}. Intuitively, we measure how long it takes for $k$ agents to visit all points in a sphere of radius $D$, as a function of $D$. Observe that such a sphere contains $\Theta(D^n)$ points, thus any algorithm having exploration cost $\Theta(D^n)$ is optimal to within a constant factor. In the synchronous model, this measure is simply the overall time taken by the robots. In the semi-synchronous model, the adversarial scheduler might schedule only one robot in each time step. In addition, the robot scheduled at a particular time step might be waiting to meet another robot and doesn't have to move. Thus, if we count the overall time taken by the robots, the adversary can make it as large as it desires. The more reasonable notion of the exploration cost in the semi-synchronous case is the total distance travelled by all robots required to visit all points in a sphere of radius $D$. Now that we have discussed this subtlety, we will abuse the terminology and use ``exploration cost'' and ``time'' interchangeably.
\section{Exploring $3$-dimensional Grids using $4$ Semi-Synchronous agents}
\label{sec:3dim}
In this section $e_1, e_2,e_3$ denote the vectors $(1,0,0), (0,1,0), (0,0,1)$, respectively. 

The  basic building block of the algorithm of \cite{EmekLSUW15} for the exploration of the 2D-grid with $4$ semi-synchronous agents  is the exploration of the perimeter of a right isosceles triangle containing the origin of the plane. The tip of the triangle is at distance $q$ from the origin with the shorter sides being the diagonals containing $2q+1$ vertices of the grid. Three of the agents are used to mark the vertices of the triangle, and the fourth one does the exploration of the sides. The value of $q$ is increased when the exploration of the triangle is finished, and the exploration of the perimeter of the larger triangle is done until the treasure is found.

%
Our algorithm for the $3D$ grids explores the sphere consisting of points at distance $q$ from the origin. In the Manhattan metric, these points are located on the triangular faces of a regular octahedron whose edges contain $q+1$ grid
vertices, see Figure \ref{fig:octo}.
\begin{figure}[h]
  \centerline{\includegraphics[width=8cm]{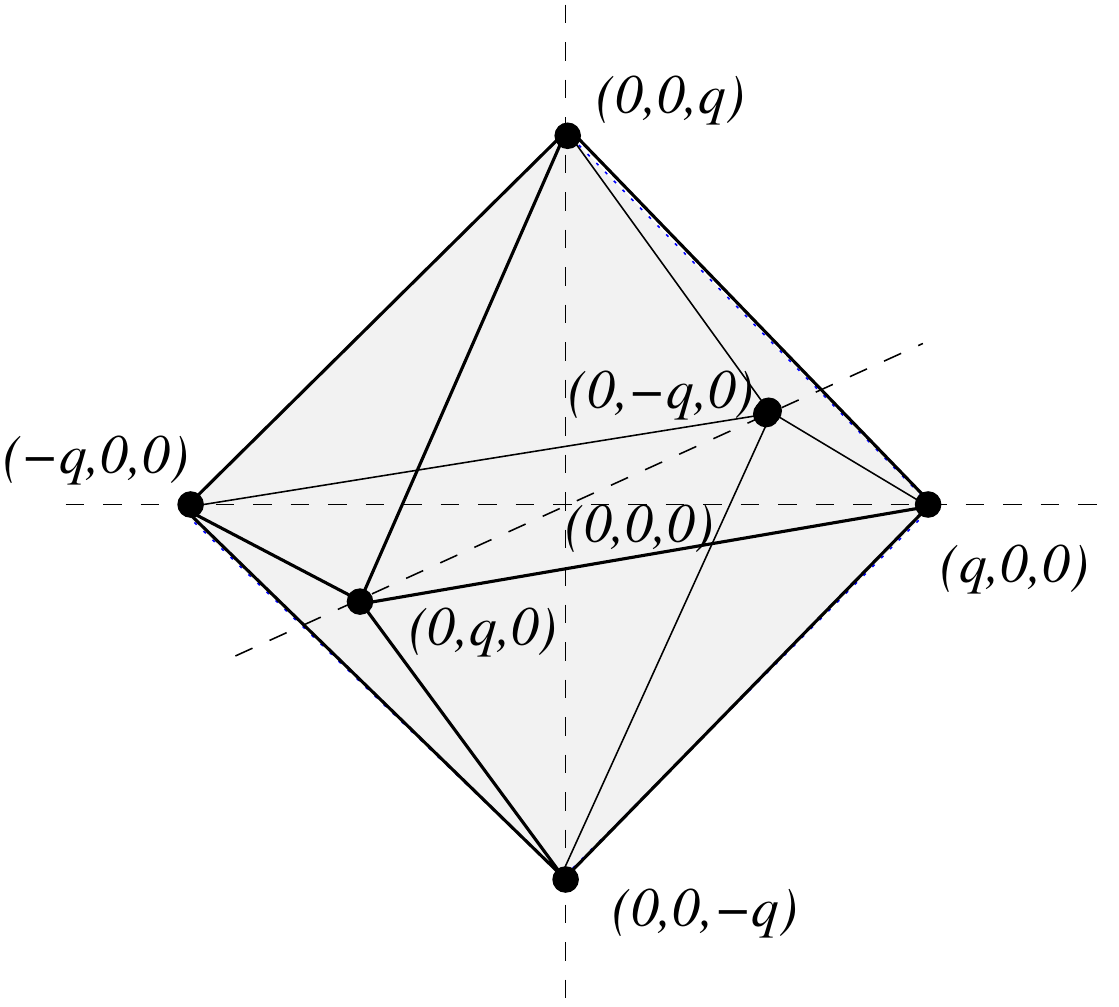}}
  \caption{Sphere of points at distance $q$ from the origin in the Manhattan metric.}
  \label{fig:octo}
\end{figure}
Thus, the basic building block of our algorithm is the exploration of all grid points on 
the {\em surface} of the equilateral triangle with vertices $r_1qe_1$, $r_2qe_2$ and $r_3qe_3$, where $r_i$'s are from $\{-1,1\}$.  
The key to our success is an algorithm for exploring one such triangle using four agents, so that\\
$-$ the value of $q$ is maintained by the distance between some of the agents while exploring a triangle, so that it can be used for the exploration of all triangles of the octahedron, \\
$-$ the exploration of all eight triangles can be done in a fixed order, and \\ 
$-$ the value of $q$ can be increased for the exploration of the larger sphere after the exploration of the sphere of radius $q$ is finished. 

We use the four agents as follows:
\begin{itemize}
\item $a$ is the {\em active} agent doing the exploration,
\item $b$ is the {\em base} agent that remains stationary during the exploration of a triangle, and
\item $c$ and $d$ mark the ends of a line segment to be explored.
\end{itemize}
The exploration of triangle  $r_iqe_i$, $r_jqe_j$, $r_kqe_k$ starts with three agents $a$, $b$ and $c$ located in $r_iqe_i$ and $d$ in node $r_jqe_j$
(see the leftmost triangle in Figure \ref{fig:3d4robots}), and ends with agent $b$ remaining in $r_iqe_i$, while the other agents are all in $r_kqe_k$. The implicit parameters $i$, $j$, $k$, $r_i$, $r_j$, $r_k$ where $i\neq j\neq k \in \{1,2,3\}$,  which control the direction of movements of agents are stored in the state of the agents. The exploration proceeds in phases; in one phase agent $a$ travels from $c$ to $d$ and back along the direction $e_j-e_i$. When  agent $a$ meets  agents
 $c$ and $d$ it pushes them one step towards $r_kqe_k$ along vectors $e_k-e_i$ and $e_k-e_j$, respectively,
as shown in the middle and rightmost triangles in Figure \ref{fig:3d4robots}.
\begin{figure}[h]
  \centerline{\includegraphics[width=15.5cm]{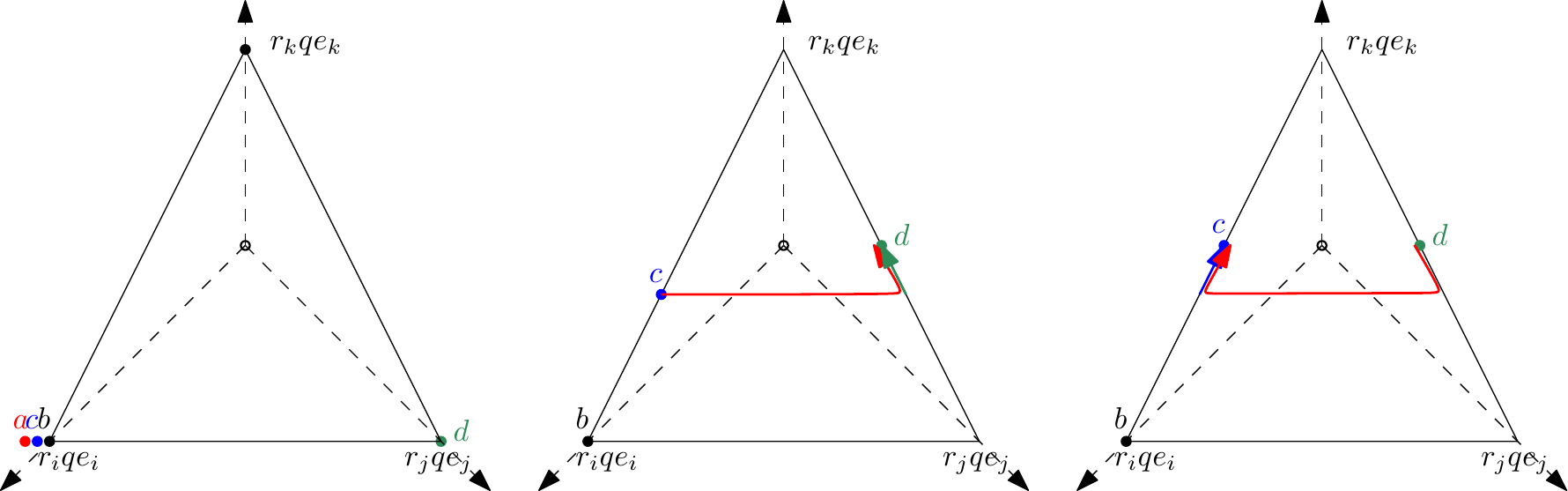}}
  \caption{Exploring a triangle using $4$ agents.}\label{fig:3d4robots}
\end{figure}

This phase is repeated until $c$ and $d$ meet at point  $r_kqe_k$, at which time the whole triangle has been explored. See the pseudocode of this exploration in Algorithm~\ref{alg:triangle}.

Observe that the value of $q$ is maintained during this exploration since $q=|bc|+|cd|$ at the beginning of each phase.  This is the crucial difference of our algorithm w.r.t. to that of \cite{EmekLSUW15}, which  keeps on increasing $q$ in order to explore the area of a triangle, making it unsuitable for  exploring the $3D$ sphere.
Furthermore, at the end of the exploration of  triangle  $r_iqe_i$, $r_jqe_j$, $r_kqe_k$, the robots are in position to start the exploration of the adjacent triangle with vertices  $r_iqe_i$ and $r_kqe_k$. It is easy to see that at the beginning of the scan of triangle  $r_iqe_i$, $r_jqe_j$, $r_kqe_k$, in the first move of agent
$a$ from $c$ to $d$ we could have moved $b$ to the location of $d$. After this modification  the robots would be in position to start the exploration of the other adjacent triangle with vertices  $r_jqe_j$ and $r_kqe_k$ at the end the the scan of the triangle. 
\begin{algorithm}[htb]
	\caption{Explore Triangle}
	\label{alg:triangle}
	\begin{algorithmic}[1]
		\State Implicit input: $q$, $i$, $j$, $k$, $r_i$, $r_j$, $r_k$
		\State On entry: $a$, $b$ and $c$ collocated at $r_iqe_i$, $d$ is at $r_jqe_j$.
		\State On exit: $b$ remains in $r_iqe_i$, $a$, $c$ and $d$ are collocated at $r_kqe_k$
		\Procedure{ExploreTriangle}{}
		\While {$c$ and $d$ do not meet}
		\State Let $d_{ik}$, $d_{jk}$ and $d_{ij}$ denote $r_ke_k-r_ie_i$, $r_ke_k-e_je_j$ and $r_je_j-r_ie_i$, respectively.
		\State $a$ goes in direction $d_{ij}$ until it meets $d$
		\State $a$ pushes $d$ on step in direction $d_{jk}$ and returns on step in direction $-d_{jk}$
		\State $a$ goes in direction $-d_{ij}$ until it meets $c$
		\State $a$ pushes $c$ one step in direction $d_{ik}$
		\EndWhile
		\EndProcedure
	\end{algorithmic}
\end{algorithm}

Exploration of the entire sphere of radius $q$ is done by doing a sequence of eight triangle explorations. Such a sequence corresponds to a Hamiltonian cycle in the $3$-dimensional hypercube, where vertices (determined by the values of $r_i$'s) represent the triangles to be explored. An edge in the Hamiltonian cycle corresponds to the transition to exploring the next triangle, and also to the edge connecting the just explored triangle with an adjacent triangle to be explored. Depending on which triangle  is explored next, the input invariant of ExploreTriangle can be made correct  by either $a$ leaving $c$ in place, or bringing $c$ to the  side of the next triangle, as pointed out  above. The order in which the faces are explored is independent of $q$ and  the agents keep the corresponding sequence of parameters $i$, $j$, $k$, $r_i$, $r_j$, $r_k$ in their states.  
Once all eights triangles have been explored, the agents are positioned back at the edge of the first triangle. At this point the value of $q$ is increased to $q+1$ by moving  agents $a$, $b$ and $c$  in direction $e_1$, and $d$ in direction $e_2$.

Thus our algorithm {\em Explore3Dgrid} simply repeats the exploration of the sphere 
of radius $q$  starting with $q=1$ and  increasing $q$ by 1 until the treasure is found.

Now we establish the exploration cost of our algorithm when it looks for the ``treasure'' located at distance $D$ from the origin. While exploring the area of a triangle with $q+1$ points at its edge, agent $a$ needs $2k+4$ steps to scan a line in the triangle containing $k+1$ points and to move agents $c$ and $d$. Thus, the exploration of a single triangle, i.e., a single facet of the octahedron, costs $\Theta(q^2)$. After that we need to reposition robots to the beginning configuration of exploring the next facet of the octahedron, and this repositioning costs at most $O(q^2)$. To explore all facets of the octahedron, we need to repeat this 7 more times. Thus, the exploration of all the facets of side length $q$ costs $\Theta(q^2)$. This needs to be repeated for $q$ from 1 to $D$, and thus our algorithm has exploration cost $\Theta(D^3)$, 
which is optimal up to a constant factor as discussed in Section~\ref{sec:model}.


Since at least four semi-synchronous agents are needed to explore a 2-dimensional grid
\cite{Brandt2018},
our result is also optimal as far as the number of semi-synchronous agents used by our algorithm is concerned. Thus we have the following theorem.
\begin{theorem} Assume that the treasure is located  in a 3D grid at distance $D$ from the origin. 
Algorithm {\em Explore3Dgrid} finds the treasure using 4 semi-synchronous agents, with the exploration cost of $O(D^3)$. 
  This is optimal as far as the number of semi-synchronous agents used, and up to a constant factor in the exploration cost.  
\end{theorem}
\section{Exploration of $n$-dimensional Grids}
\label{sec:ndim}
A straightforward generalization of the algorithms for the exploration of 2D grids  \cite{EmekLSUW15} to $n$ dimensions results in algorithms that use $\Omega (n)$ agents. Consider, for example, such a simple generalization of a randomized 2D algorithm.   
The basic idea of the $n+1$-agent randomized algorithm for $n$ dimensions is to make an $n$-segment walk, starting from the origin, and walking the $i$-th segment along dimension $i$. The lengths of the segments are chosen randomly, and one agent per segment is used to mark its endpoint. This allows the agent to find the way back to the origin and start another random trial.

In essence, this algorithm uses $2$ agents per dimension to store in unary the distance travelled in this dimension, and by an appropriate arrangement we can reuse one of the agents in the successive dimension to bring the number of additional agents per dimension to $1$.

The main idea of our approach is a realization that it is not necessary to use $n+1$ agents to store $n$ numbers of segment lengths. Observe that segment lengths are stored and retrieved in this randomized algorithm in the first-in last-out order. Thus this algorithm can be realized if we can implement a stack of the agent's movements. Turns out we can use a {\em constant} number of agents, independent of the grid's dimension, to {\em implement a stack} into which the active agent, that does the exploration, stores its walk and then it can use it to return to the origin. The active agent carries the stack along its walk. 

\subsection{The Stack Implementation}
The format of data stored in the logical stack is the string $\alpha \in (0^*1)^n$, where $0$ represents {\em continue walking in the current direction}, $1$ represents {\em switch to the next dimension}. \\    
The physical implementation of the stack stores this data by interpreting $\alpha^r$ (that is $\alpha$ reversed) as a binary number $S$ and storing it in unary as a distance between two agents located in a row in the first dimension.

We employ the following agents:
\begin{itemize}
	\item $a$: the {\em active} agent that is doing the exploration of the grid; in the semi-synchronous model this is the only agent moving around and manipulating the other agents,
	\item $b$: the {\em base} of the stack, from which measurements are taken, and representing the current logical location of the exploration,
	\item $c$: the {\em counter} agent; this is an auxiliary agent for implementing the stack operations in the semi-synchronous model,
	\item $d$: the {\em distance} agent; its distance from the base $b$ stores the content of the stack,
	\item $e$: the {\em extra} agent used in the deterministic algorithms to store an extra copy of the current stack value.
\end{itemize}
The basic stack operations we need to implement are {\em isEmpty()}, {\em push(v)} where $v\in\{0,1\}$ and {\em pop()}. Operation 
{\em isEmpty()} simply returns whether $b$ and $d$ are collocated. Implementation of {\em push()} and {\em pop()} is model-dependent and given below. 

\subsubsection{Implementing Semi-Synchronous Stack}

Algorithms~\ref{alg:asynchPush} and \ref{alg:asynchPop} show the implementation of push and pop operations for the semi-synchronous stack.

\begin{algorithm}[htb]
	\caption{Semi-synchronous stack: push(v)}
	\label{alg:asynchPush}
	\begin{algorithmic}[1]
		\State On entry: $b$ and $c$ collocated, $a$ and $d$ collocated at $b+Se_1$.
		\State On exit: $b$ and $c$ collocated, $a$ and $d$ collocated at $b+(2S+v)e_1$.
		\Procedure{push}{$v$}
		\State $a$ goes to $b$ and brings $c$ to $d$
		\While {$b$ and $d$ are not collocated}
				\State $a$ goes to $c$, pushes it one step away from $b$ and returns to $d$ 
				\State $a$ pushes $d$ one step closer to $b$
		\EndWhile
		\State $d$ becomes $c$
		\State $a$ goes to $c$ and tells it to become $d$
		\If {$v$=1}
			\State $a$ pushes $d$ one step away from $b$
		\EndIf
		\EndProcedure
	\end{algorithmic}
\end{algorithm}

\begin{algorithm}[htb]
	\caption{Semi-synchronous stack: pop()}
	\label{alg:asynchPop}
	\begin{algorithmic}[1]
		\State On entry: $b$ and $c$ collocated, $a$ and $d$ collocated at $b+Se_1$.
		\State On exit: $b$ and $c$ collocated, $a$ and $d$ collocated at $b+\lfloor S/2\rfloor e_1$, returns $S \bmod 2 = 1$.
		\Procedure{pop}{}
		\While {$b$ and $c$ are at distance more than $1$}
				\State $a$ pushes $d$ one step closer to $b$
				\State $a$ goes to $c$ and pushes it one step away from $b$
		\EndWhile
		\State $v = d$ is one step from $c$
		\State $c$ and $d$ switch roles
		\State $a$ brings $c$ to $a$ and returns to $d$
		\State return $v$
		\EndProcedure
	\end{algorithmic}
\end{algorithm}

\subsubsection{Implementing Synchronous Stack}
In the synchronous model, we can synchronize the movements of agents  to effectively multiply or divide the stack content by $2$ without the need of the counter agent $c$,  see Figure \ref{fig:synchStack}.

\begin{algorithm}[htb]
	\caption{Synchronous stack: push(v)}
	\label{alg:synchPush}
	\begin{algorithmic}[1]
		\State On entry: $a$ and $d$ collocated at $b+Se_1$.
		\State On exit: $a$ and $d$ collocated at $b+(2S+v)e_1$.
		\Procedure{push}{$v$}
		\State $a$ goes to $b$ and then back towards $d$ until they meet, walking at full speed of $1$
		\State $d$ walks away from $b$ at speed $3$ (move, wait, wait, see Figure \ref {fig:synchStack})
		\If {$v$=1}
			\State $a$ pushes $d$ one step away from $b$
		\EndIf
		\EndProcedure
	\end{algorithmic}
\end{algorithm}

\begin{algorithm}[htb]
	\caption{Synchronous stack: pop()}
	\label{alg:synchPop}
	\begin{algorithmic}[1]
		\State On entry: $a$ and $d$ collocated at $b+Se_1$.
		\State On exit: $a$ and $d$ collocated at $b+\lfloor S/2\rfloor e_1$, returns $S \bmod 2 = 1$.
		\Procedure{pop}{}
		\State $a$ goes to $b$ and then back towards $d$ until they meet, walking at full speed of $1$
		\State $d$ walks towards $b$ at speed $3$ (move, wait, wait, see Figure \ref {fig:synchStack}) 
		\If {$a$ and $d$ meet right after $d$'s move}
			\State return $1$
		\Else
			\State return $0$
		\EndIf
		\EndProcedure
	\end{algorithmic}
\end{algorithm}
\begin{figure}[h]
  \centerline{\includegraphics[width=16cm]{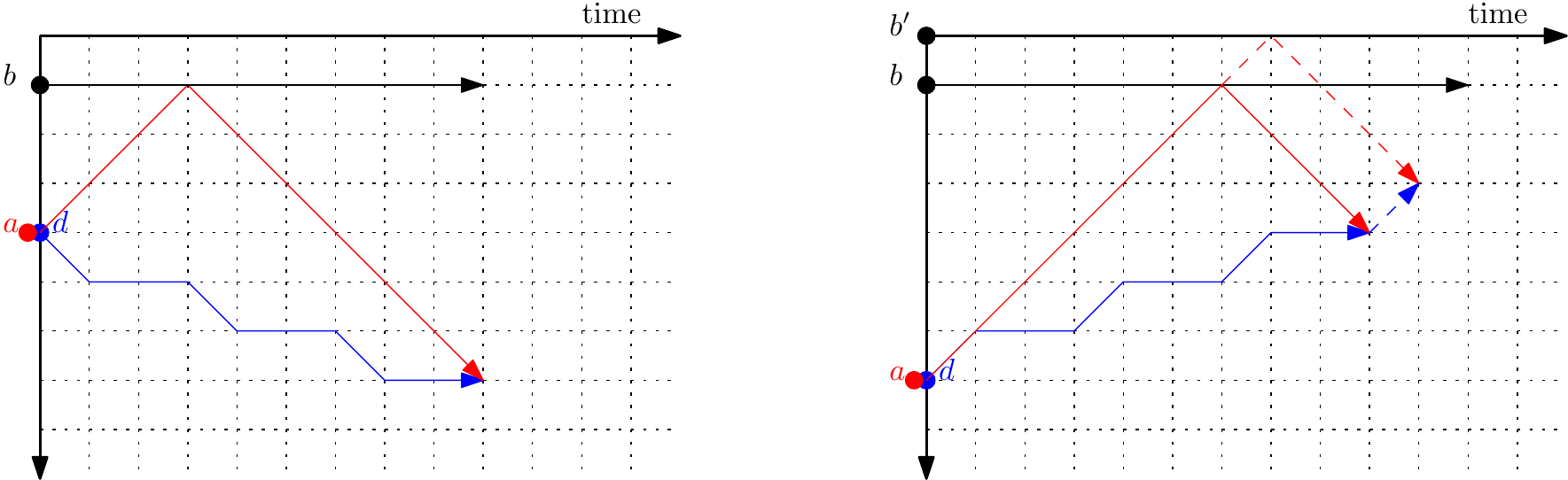}}
  \caption{Implementation of multiply (left) and divide (right) using synchronous agents. Both odd and even cases are shown for divide.}\label{fig:synchStack}
\end{figure}

\subsection{The Randomized Algorithm}
\label{SSec:nRandom}
As already stated in the initial part of this section, the main idea of the algorithm is to use the stack to store the random choices during the walk, so that the agent can return to the origin. The agent $a$ carries the stack along this walk so that the operations can be applied without the need to search for the stack.

In addition to the stack methods, it uses two new procedures. Procedure {\em random($p$)} returns $1$ with probability $p$, while {\em moveStack()} moves the whole stack one step in the direction specified. Note that since the whole stack is located on a single line, the agent $a$ can do that with finite memory.

\begin{algorithm}[h]
	\caption{Randomized Grid Exploration}
	\label{alg:randAlg}
	\begin{algorithmic}[1]
		\While {treasure not found}
		\State Pick a random $n$-bit string $R\in\{-1,1\}^n$
		\For {$i=1$ to $n$ }
			\While {random$(p)=0$}
				\State push($0$)
				\State moveStack($r_ie_i$)
			\EndWhile
			\If {$i<n$}
				\State push($1$)
			\EndIf
		\EndFor
		\State $i=n$
		\While {not empty()}
			\While {pop()$=0$}
				\State moveStack($-r_ie_i$)
			\EndWhile
			\State $i=i-1$
		\EndWhile
		\EndWhile
	\end{algorithmic}
\end{algorithm}

The algorithm works in rounds, which we number $1, 2, 3, \ldots$, which correspond to the iteration numbers of the outer while loop. At the beginning of each round, the active robot picks a binary string $R \in \{-1,1\}^n$ uniformly at random. This string indicates that the robot is going to explore dimension $i$ in direction $R_i$. Then for each dimension $i$ from $1$ to $n$, the active robot travels for $Z_i-1$ steps in direction $R_i$, where $Z_i$ is geometrically distributed with parameter $p$ (to be determined later). Note that we want $Z_i$ to represent the length of the string pushed onto the stack while moving in dimension $i$. Since the string pushed on the stack includes the ``separator'' between dimensions, we have the $-1$ term for the actual number of moves. We call the concatenation of all such moves over all dimensions \emph{the logical path of the active robot}. If no treasure is found, the active robot uses the stack to retrace its logical path back to the origin by travelling $Z_n-1$ steps in direction $-R_n e_n$ first, followed by $Z_{n-1}-1$ steps in direction $-R_{n-1} e_{n-1}$, and so on. To estimate the exploration cost of each round, we need a simple helper lemma.

\begin{lemma}\label{lm:oneWalkCost}
Let $S$ be the maximal stack size during one iteration of the outer while loop of Algorithm~\ref{alg:randAlg}. 
The overall cost of this iteration 
is $O(S^2)$ when implemented by semi-synchronous agents, and is  $O(S)$ when implemented by synchronous agents.
\end{lemma}
\begin{proof}
In the semi-synchronous model, each push() or pop() costs $O(X^2)$, where $X$ is the actual stack size, as the active agent zig-zags between $b$ and $d$.
On the other hand, in the synchronous model, the cost of each operation is linear in the stack size.
The cost of moving the stack is linear in both models.

As the stack size grows exponentially, and then reduces exponentially, the overall cost is determined by the cost when the stack is the largest, i.e. $O(S^2)$ and $O(S)$ for the semi-synchronous and synchronous models, respectively.
\end{proof}

Observe that during a given round the maximum size of the stack is $2^{Z_1 + \cdots + Z_n}$. Thus the exploration cost of each round is at most $2 (Z_1 + \cdots +Z_n) 2^{\Theta(Z_1 + \cdots + Z_n)}$, where $2(Z_1 + \cdots + Z_n)$ is the bound on the overall length of the logical path (there and back) of the active robot, and by Lemma~\ref{lm:oneWalkCost} each step of the active path costs $2^{\Theta(Z_1 + \cdots + Z_n)}$, since we need to perform operations on the stack of size $2^{Z_1 + \cdots + Z_n}$. Also note that $2 (Z_1 + \cdots +Z_n) 2^{\Theta(Z_1 + \cdots +Z_n)} = 2^{\Theta(Z_1 + \cdots +Z_n)}$. Let $c$ be the constant in the $\Theta$ notation such that the exploration cost of a round is at most $2^{c(Z_1 + \cdots +Z_n)}$.

For simplicity, we will assume that the active robot checks for the treasure only at the far end-point of the logical path in each round. This assumption might lead to a more pessimistic upper bound on the exploration cost than if  we assumed that the active robot checks for treasure at each grid point that it visits. However, our assumption simplifies the calculations and is sufficient for our purposes.

\begin{theorem} \label{thm:rand} Algorithm~\ref{alg:randAlg} locates the treasure in the $n$-dimensional grid in finite expected time, using either $4$ semi-synchronous or $3$ synchronous agents.
\end{theorem}
\begin{proof} 
Consider the infinite sequence of random variables $(X_i)_{i=1}^\infty$, where $X_i$ is the exploration cost of round $i$. Note that the $X_i$ are independent and identically distributed. Consider the exploration cost of a particular round, e.g., $X_1$. Then we have $X_1 \le 2^{c(Z_1 + \cdots +Z_n)}$, where the $Z_i$ and $c$ are as defined above. Then we have
\begin{align*}
\mathbb{E}(X_1) &\le \mathbb{E}\left(2^{c(Z_1 + \cdots +Z_n)}\right) \\
&= \sum_{i_1 = 1}^\infty \sum_{i_2 = 1}^\infty \cdots \sum_{i_n = 1}^\infty 2^{c(i_1 + \cdots+ i_n)} p^{i_1-1} (1-p) p^{i_2-1} (1-p) \cdots p^{i_n-1} (1-p)\\
&= \left(\sum_{i_1=1}^\infty (2^c p)^{i_1-1} 2^c(1-p)\right) \left( \sum_{i_2 = 1}^\infty (2^c p)^{i_2-1} 2^c(1-p)\right)\cdots \left(\sum_{i_n=1}^\infty (2^c p)^{i_n-1} 2^c(1-p)\right)\\
&=2^{cn}(1-p)^n \frac{1}{(1-2^c p)^n},
\end{align*} 
where the last step holds as long as $2^c p < 1$ that is $p < 1/2^c$.

Define a random variable $T$ to be the minimum $t$ such that the far end-point of $X_t$ coincides with the treasure. That is our exploration procedure terminates in round $T$, but not earlier. Suppose that the treasure is located at position $(k_1, \ldots, k_n)$ where $|k_1| + \cdots + |k_n| = D$. By the discussion immediately preceding the statement of this theorem, the probability that the treasure is found in a particular round is $\widehat{p} = 2^{-n} (1-p)^n p^{k_1} \cdots p^{k_n} = 2^{-n}(1-p)^n p^D$, where $2^{-n}$ is the probability of guessing correctly the signs of the $k_i$ and $p^{k_i} (1-p)$ is the probability of travelling the correct number of steps in dimension $i$. Thus $T$ is geometrically distributed with parameter $\widehat{p}$. Therefore, $\mathbb{E}(T) = 1/\widehat{p}$.

We are interested in bounding the overall exploration cost, that is $\mathbb{E}(X_1 + \cdots + X_T)$. Since the $X_i$ are i.i.d. and $T$ is a stopping time, it follows by a generalization of the Wald's equation~\cite{wald1944} to stopping times that
\[ \mathbb{E}(X_1 + \cdots + X_T) = \mathbb{E}(T) \mathbb{E}(X_1) \le \frac{1}{\widehat{p}} 2^{cn}(1-p)^n \frac{1}{(1-2^cp)^n} < \infty.\]
This holds as long as we choose $p < 1/2^c$. Since $c$ is a constant, such a probabilistic coin can be implemented by finite state machines. The statement of the theorem follows by the number of robots sufficient to implement stack operations in each of the models (synchronous vs. semi-synchronous).
\end{proof}

\subsection{The Deterministic Algorithm}
\label{SSec:det}
The main idea is to exhaustively go over all possible stack contents in increasing order, interpreting each stack as a specification of a walk. We also keep a backup of the initial stack content, and at the end of the walk we use the backup  to return to the origin. 
The back-up stack  is stored using an additional agent. The backup is needed, as reading the stack content during the walk destroys it. Note that after the outward walk, we do not logically reverse the stack; hence the return to the origin does not use the same path as the original walk. However, this is not a problem as the walks along different dimensions are commutative.

Finally, we should mention that some generated stacks do not necessarily have the correct format, some may contain too few or too many $1$s. However, this is easy to handle by the algorithm: too few ones just means we walked without using all of the  dimensions, which is still a perfectly valid walk. The excessive $1$s are simply ignored by taking the first excessive $1$ as a directive to end the walk and return to the origin. 

\begin{algorithm}[htb]
	\caption{Deterministic Grid Exploration}
	\label{alg:detAlg}
	\begin{algorithmic}[1]
		\While {treasure not found}
		\State Increment the backup stack
		\For {every $n$-bit string $R\in\{-1,1\}^n$}
		\State execute Walk($R$, $1$)
		\State execute Walk($R$, $-1$)
		\EndFor
		\EndWhile
		\State $ $	
		\Procedure{Walk}{$R$, $s$}
		\State Restore stack from backup
		\State $i=1$
		\While {not empty() and $i\leq n$}
			\While {pop()$=0$}
				\State moveStack($sr_ie_i$)
			\EndWhile
			\State $i=i+1$
		\EndWhile
		\EndProcedure
	\end{algorithmic}
\end{algorithm}

Using essentially the same arguments as in Lemma~\ref{lm:oneWalkCost} yields
\begin{lemma}\label{lm:procWalkCost}
The cost of procedure Walk is $O(S^2)$ and $O(S)$ in the semi-synchronous and synchronous models, respectively, where $S$ is the size of the backup stack.
\end{lemma}
\begin{theorem} \label{thm:det} Algorithm~\ref{alg:detAlg} locates the treasure in the $n$-dimensional grid with:\\
\hspace*{4mm} $5$ agents and  the exploration cost of $O(2^{3D+4n})$ moves in the semi-synchronous model, and \\
\hspace*{4mm} $4$ agents  and the exploration cost of $O(2^{D+2n})$ in the synchronous model.
\end{theorem}
\begin{proof}
The number of agents and the correctness follows easily from the construction.

It remains to sum up the cost of all calls to procedure Walk. Note that each point in space uniquely specifies a valid (i.e. with precisely $n$ $1$'s) stack.
Hence, the valid stack for the treasure at distance $D$ contains $D+n$ digits. Therefore, the overall cost of Algorithm~\ref{alg:detAlg} is 
$$2^n \sum_{X=1}^{2^{D+n}} O(X^2) = O(2^n(2^{D+n})^3) = 2^{3D+4n}$$ in the semi-synchronous model, and $O(2^{D+2n})$ in the synchronous model (the initial $2^n$ covers all choices for string $R$).
\end{proof}

\section{Polynomial time solutions}
\label{sec:ndimpol}

While designing our exploration algorithms in the previous section, we concentrated on minimizing the number of agents used, and the resulting cost of these algorithms is exponential in the volume $V(D)$, the smallest ball containing the treasure. A natural question to ask is whether this is an unavoidable consequence of using only a constant number of agents in the exploration.   
In this  section we show  that  this is not the case: a single additional agent is sufficient to bring the cost of exploration down to a polynomial in $V(D)$. 

The main reason the cost of algorithms in the preceding section is exponential is the number of incorrect stack contents being considered:  as $D$ grows compared to the fixed $n$, ever larger proportion of stack contents does not have the correct format and they result in repeatedly reaching already explored vertices. To avoid this problem we will efficiently explore an $n$-dimensional cube $q^n$ of side $q$ centered at the origin. We use again the stack idea to trace the exploration of $q^n$. The logical stack content now consists of $n$ numbers in $q$-ary alphabet, describing a location within this cube. However, in this case, we also need to store the scale $q$.  As before, the stack implementation interprets the logical content as a $q$-ary number and stores it in unary\footnote{This is similar to the simulation of PDAs by counter machines — see Chapter 8.5 in Hopcroft, Motwani, and Ullman text~\cite{HopcroftMU2006}; however, the details of our implementation are completely different.}. Since $q$ also needs to be stored on its own, this incurs the additional cost of one agent. However, this allows us to multiply and divide by $q$, which would not have been possible without the extra agent.

 The stack is manipulated using the explicit commands:\\
- isDivisible() which checks the divisibility by $q$,\\
- push($0$) which multiplies the stack content  by $q$,\\
- pop() which divides the stack content by $q$, and\\
- increment() which increments the top of the stack.

\subsection{Stack operations: semi-synchronous  implementation}
In addition to agents $a$, $b$ and $d$, we use agent $f$ to  maintain the value of $q$ by  placing it at $b+qe_1$.
Furthermore, two counter agents $c_d$ and $c_f$ are used. At the beginning of the stack operations, $f$ and $c_f$ are collocated, as are $b$ and $c_d$, and 
$a$ and $d$. The basic procedure is a traversal of the whole stack by agent $a$, manipulating the tokens according to the specific command.

In push($0$) (i.e. multiplying the stack content by $q$), $a$ pushes $c_f$ towards $b$ and $c_d$ away from $b$. Whenever $c_f$ reaches $b$, a transports it back to $f$ as well as pushes $d$ one step closer to $b$. The process terminates when $d$ reaches $b$; subsequently $c_d$ and $d$ change roles. The detailed procedure is given in Algorithm~\ref{alg:asynchPush-multiplyq}. It is easy to see that the outer loop executes $S$ times, where $S$ is the size of the stack at the start of the algorithm, and the inner loop $q$ times, and each iteration of the inner loop takes at most $Sq$ steps. Thus, the total cost of the push($0$) operation is bounded by $O({S'}^2)$ where $S'=Xq$ is the size of the stack at the end.

In isDivisible(), $a$ pushes $c_f$ towards $b$ and $c_d$ towards $d$, until $c_d$ reaches $d$. Whenever $c_f$ arrives to $b$, it is transported back to $f$. isDivisible() returns true iff at the moment when $c_d$ reaches $d$, $c_f$ is at $b$ (or $f$).

pop() means dividing the stack by $q$. The process is essentially reverse of push() -- in every iteration/traversal of the stack,  $c_f$ and $d$ are pushed towards $b$. Whenever $c_f$ reaches $b$, it is brought back to $f$ and $c_d$ is pushed away from $b$. When $d$ reaches $b$, $c_d$ and $d$ exchange their roles.

The detailed pseudocode of isDivisible() and pop() are straightforward and omitted. 

\begin{algorithm}[htb]
	\caption{Semi-synchronous stack implementation of multiplication by $q$: push(0)}
	\label{alg:asynchPush-multiplyq}
	\begin{algorithmic}[1]
		\State On entry: $b$ and $c_d$ collocated, $a$ and $d$ collocated at $b+Se_1$, and $f$ and $c_f$  collocated at $q$
		\State On exit: $b$ and $c_d$ collocated, $a$ and $d$ collocated at $b+qSe_1$, and $f$ and $c_f$  collocated at $q$.
				\Procedure{push}{$0$}
		\While {$b$ and $d$ are not collocated}
			\State $a$ pushes $d$ one step closer to $b$
		\While {$c_f$ and $b$ are not collocated}
				\State $a$ goes to $c_d$ and pushes it one step away from $b$
				\State $a$ goes to $c_f$ and pushes it one step  towards $b$ 
		\EndWhile.
		        \State $a$ brings $c_f$ to $f$.
		
		\EndWhile.
		\State $d$ becomes $c$
		\State $a$ goes to $c$ and tells it to become $d$

		\EndProcedure
	\end{algorithmic}
\end{algorithm}

\subsection{Stack operations: Synchronous  implementation}
A straightforward application of the technique from Section \ref{sec:ndim} would need agents traveling at speed $\frac{1}{2q+1}$ (for multiply) and $\frac{q-1}{s+q}$ (for divide), which is impossible with finite state agents. 

Instead, we take  $q$ to be a power of two and implement the operation of multiply, divide by $q$ via  repeated applications of multiplication by $2$, division by $2$, respectively. 
Thus in this case $f$ is placed at distance $\log q$ from $b$, instead of placing it at distance of $q$ from $b$. The counter $c_f$ is used to count the number of multiplications/divisions already performed, while the counter $c_d$ is not used  at all, i.e. only agents $a$, $b$, $d$, $f$ and $c_f$ are needed. The operations of doubling and halving were already described in Section~\ref{sec:ndim} and shown to take $O(S)$ time. Since these operations are performed $\log q$ times, the total time complexity of every stack operation is $O(S \log q)$. 

\subsection{Fast deterministic grid exploration}

Our  polytime deterministic grid exploration algorithm is described in Algorithm~\ref{alg:fastDetAlg}. Starting with $q=2$, and for any fixed value of $q$, the algorithm generates and visits the addresses ($n$-tuples from a $q$-ary alphabet) in lexicographic order. Then the agent $a$  moves to position $(-q, -q, \ldots, -q)$, doubles the value of $q$, and moves on to the next iteration. Agent $a$ always drags the stack along as it performs the exploration. The procedure $explore(i)$  is a recursive procedure to generate $n$-tuples in lexicographic order; it is called with logical  stack content  an $i$-tuple $x0$. It then iteratively calls $explore(i+1)$ to visit the $(n-i)$-dimensional cube of side $q$ with $(x,j,0, \ldots, 0)$ as the origin, for $j$ ranging from $0$ to $q-1$. 

Note that the algorithm as shown in Algorithm~\ref{alg:fastDetAlg}  is presented using recursive calls  for convenience; however, $i$ is maintained in the local state. 

\begin{theorem} \label{thm:fast}  Let $V(D)$ be the volume of the ball of diameter $D$ in the $n$-dimensional grid. Algorithm~\ref{alg:fastDetAlg} locates the treasure in the $n$-dimensional grid with:\\
\hspace*{4mm}  $6$ agents and the exploration cost of $O(V(D)^3))$ moves in the semi-synchronous model, and\\
 \hspace*{4mm} $5$ agents and the exploration cost of $O(V(D)^2\log D)$ In the synchronous model.
\end{theorem}
\begin{proof}
The number of agents and the correctness follows easily from the construction.

It remains to sum up the cost of all stack operations on a stack of size $S$. As already described,  the cost of each stack operations is $O(S^2)$ and $O(S\log{q})$ in the semi-synchronous and synchronous models, respectively. The maximal stack size $S$ is bound by $q^n$, which is also the number of points covered by the stack base during one iteration of the outer loop (i.e. for fixed $q$). This results in the overall cost of $O(S^3)$ and $O(S^2\log q)$ in the semi-/pol	synchronous and synchronous models, respectively. As $q$ grows exponentially, the overall cost is determined by the cost for the last value of $q$.

Finally, it is known that $V(D) = \frac{2^n}{n!}D^n$. As $q<4D$ (the treasure would had been found if $q\geq 2D$), we get that $S\leq (4D)^n = 2^nn!V(D)$, where
 $n$ is a constant.  This proves the theorem.
\end{proof}

\begin{algorithm}[htb]
	\caption{Fast Deterministic Grid Exploration}
	\label{alg:fastDetAlg}
	\begin{algorithmic}[1]
	\State $q=2$
	\State push($0$)
	\While {treasure not found}
		\State explore($1$)
		\State moveStack($-q\sum_{i=1}^ne_i$)
				\State $q=2q$
	\EndWhile	
	\State $ $
	\Procedure{explore}{$i$}
		\If{$i>n$}
			\State return
		\EndIf
		\Repeat 
			\State push($0$)
			\State explore($i+1$)
			\State increment()
		  	\State moveStack($e_i$)
		\Until{isDivisible()}
		\State pop()
		\State moveStack($-qe_i$)
	\EndProcedure

	\end{algorithmic}
\end{algorithm}

\section{On the Size of the Explored Space}
\label{sec:bounded}
In our exploration algorithms for general $n$ in Sections \ref{sec:ndim} and \ref{sec:ndimpol}, the agents employed a stack of size exponential in $D$. 
In this section we address the question of whether such behavior is necessary, or if there are exploration algorithms, which we call space-efficient, that limit the size of the space visited during the exploration to a constant factor of $D$. In Subsection~\ref{sec:space-efficient}, we present high-level details of a space-efficient algorithm that uses more than a constant number of agents (but still $o(n)$). 
While we are unable to prove the general lower bound saying that $3$ synchronous agents cannot explore $\mathbb{Z}^3$, in Subsection~\ref{sec:lb-3-synch} we show that there is no \emph{space-efficient} algorithm with $3$ synchronous agents to explore $\mathbb{Z}^3$. More specifically, every algorithm with $3$ agents that explores all grid points within distance $D$ must have an agent travel distance $\Omega(D^{3/2})$ away from the origin at some point in time.

\subsection{Space-Efficient Exploration with Many Agents}
\label{sec:space-efficient}

The main idea for limiting the visited space is to encode the needed information in a more compact way, using more agents. Our previous solutions had a single active agent doing the exploration, and a constant number of agents that implement a stack that stores several numbers encoded as a single number and represented as a distance between two agents. Now, instead of a stack, the active agent carries around a $\sqrt{n}$-dimensional sub-cube of side length $q$. A non-active agent inside such a sub-cube can be used to represent $\sqrt{n}$ base numbers from $\{0, 1, \ldots, q-1\}$ --- consider simply the coordinates of the non-active agent relative to the origin of the sub-cube. Therefore, $\sqrt{n}$ agents can be used to represent $n$ numbers from $\{0,1, \ldots, q-1\}$. When these numbers are juxtaposed they correspond to a single $n$-digit base $q$ number --- the coordinate of the grid point that is currently being explored by the active agent inside $\{0, \ldots, q-1\}^n$. The active agent needs to be able to explore the sub-cube, reorganize all non-active agents inside the sub-cube to point to the next $n$-digit base $q$ number, and move itself and the entire sub-cube to the new location indicated by the updated positions of non-active agents. The active agent needs to be able to reorganize all non-active agents in such a way as to enumerate all possible $n$-digit base $q$ numbers. Once that happens, the agent can run the protocol in reverse, return to the origin, increment $q$ and repeat the process.

Although the details are tedious and omitted in this version of the paper, one can easily verify that the active agent can perform the operations required to enumerate all $n$-digit base $q$ numbers: increment a number, which means push the corresponding non-active agent inside the cube by 1 along one of the $\sqrt{n}$ axis; check whether the number has reached $q$, which means check if some non-active agent is on some facet of the sub-cube; set the number to zero, which means bring a non-active agent located on some facet of the cube to the opposite facet of the cube; and proceed to the next/previous digit, which corresponds to modifying the internal state of the active agent. These tasks might require extra agents, but we can always employ a simple algorithm using $\sqrt{n}+2$ agents to explore an $\sqrt{n}$-ary sub-cube of side $q$ without ever leaving the sub-cube. Altogether $2 \sqrt{n}+O(1)$ agents are sufficient to implement this entire scheme. The overall exploration cost is $O(q^{n+\sqrt{n}})$ as visiting each node incurs the overhead of traversing the whole memory of size $q^{\sqrt n}$. This improves upon the results from the previous section in terms of exploration cost, and simultaneously limits the exploration to points at distance at most $2D$ from the origin.

This technique can be applied recursively: Let $n_i$ denote the number of dimensions we can explore at logical level $i$ using $g_i$ agents. For level $i+1$, we can use $f_i$ agents in $n_i$-dimensional space to encode $f_in_i$ numbers, yielding $n_{i+1} = n_if_i$ and $g_{i+1} = g_i+f_i$. If we try to minimize $g_i$ w.r.t. to $n_i$, we will choose $f_i=2$ for all $i$, resulting in $g_i\in \log n_i$, i.e. $O(\log n)$ agents are sufficient to explore $n$-dimensional grids while limiting the visited space to $O(D)$.




\subsection{Lower Bound on the Visited Space for $3$ Synchronous Agents for $n=3$}
\label{sec:lb-3-synch}

Using the techniques and results from~\cite{EmekLSUW15}, it is possible to show that for any distance $d$ there are only $O(d)$ configurations in which two agents are collocated and the third one is at distance $d$. Furthermore, again based on previous results we know that there must be infinitely many meetings between pairs of agents. As the agents are finite automata moving (when looking from sufficiently far above) in straight lines, the number of explored vertices between two consecutive meetings is $O(d)$. Combining with the fact that there are $\Omega(D^3)$ grid points in the ball of radius $D$ and the fact that the number of meetings at distance less than $d$ is $O(d)$ (and hence, their total contribution to the number of explored nodes is at most $O(d^2)$) means that in order to visit all vertices in the ball of radius $D$ the distance between agents must have been $\Omega(D^{3/2})$ at some moment before locating the treasure. In what follows, we give formal arguments supporting the above intuition.

Suppose that we have an algorithm that uses several agents to find a treasure at an unknown location. Observe that if we run such an algorithm on an empty grid, i.e., without a treasure at all, then eventually every grid point has to be visited by some agent -- this is an equivalent view of a treasure search algorithm. Throughout this section we will often adopt this point of view and think of a treasure search algorithm as running on an empty grid and having to ``cover'' all grid points eventually.

We first start with a few general definitions and lemmas that apply to any dimension $n$ of the ambient space.

\begin{definition}
A cylinder in direction $v \in \mathbb{R}^n$ of radius $r$  from origin $x_0 \in \mathbb{R}^n$ is the set
\[ \{ x\in \mathbb{R}^n \mid \exists t \in \mathbb{R} \text{ such that } ||x-(t v+x_0)||_{\infty} \le r\}.\]
\end{definition}

Intuitively, between meeting each other agents move along cylinders. Precise statements follow below, but for now we make an easy observation that no finite number of cylinders can cover all grid points.

\begin{lemma}\label{lem:finitely-many-cylinders}
Finitely many cylinders cannot cover all of $\mathbb{Z}^n$.
\end{lemma}
\begin{proof}
Consider cylinders $C_1, \ldots, C_k$ with radii $r_1, \ldots, r_k$. Now consider a ball of radius $R$. The number of integral points within the ball is $\Theta(R^n)$. The number of integral points within  the ball that are covered by cylinder $C_i$ is at most $\Theta(r_i^{n-1} R)$ (the height of the cylinder relevant to the ball is at most $2R$). Assume for contradiction that the cylinders cover all of $\mathbb{Z}^n$, then they cover all the integral points within the ball as well. Thus, we must have
\[ \Theta\left( \sum_{i=1}^k r_i^{n-1} R\right) \ge \Theta(R^n).\]
The left hand side is a linear function of $R$, while the right hand side is a polynomial of degree $n$ of $R$. Thus, a large enough value of  $R$  would violate the inequality. This leads to a contradiction.
\end{proof}

The following definition makes precise the local view of the world by a set of agents.

\begin{definition}
Consider a set of agents $\{a_1, \ldots, a_k\}$ in $\mathbb{R}^n$ at a particular time $t$. The agents are at positions $(x_1, \ldots, x_k)$ and in states $(q_1, \ldots, q_k)$. The tuple $(x_1, \ldots, x_k, q_1, \ldots, q_k)$ is called the \emph{configuration} of the agents at time $t$. The tuple $(y_1, \ldots, y_k, q_1, \ldots, q_k)$ where $y_i = x_i - x_1$ is called the \emph{relative configuration} of the agents at time $t$. Note that to obtain the relative configuration we simply shift the origin of the coordinate system to agent $a_1$.
\end{definition}

The following is the main helper lemma that will be used multiple times to establish the precise behavior of 3 agents in $\mathbb{R}^3$. It relates repeating relative configurations to cylinders.

\begin{lemma}\label{lem:rel-conf-helper}
Consider agents $a_1, \ldots, a_k$ exploring $\mathbb{Z}^n$. Suppose that the agents interact only with each other and no other agents from some time $t$ onward. If the relative configuration repeats then all the grid points visited by the agents fall within a cylinder of radius $r$ and direction $v$, where $r$ and $v$ depend only on the original relative configuration.
\end{lemma}
\begin{proof}
Let $conf_i = (x_1^i, \ldots, x_k^i, q_1^i, \ldots, q_k^i)$ denote the absolute configuration at time $i$, and let $rconf_i = (y_1^i, \ldots, y_k^i, q_1^i, \ldots, q_k^i)$ denote the relative configuration at time $i$. Consider the $i_1 < i_2$ such that $rconf_{i_1} = rconf_{i_2}$. Since the agents are deterministic finite state automata, they are going to repeat exactly the same sequence of steps from $rconf_{i_2}$ as they did from $rconf_{i_1}$ --- the relative configuration corresponds to the view of the world as perceived by the agents. Thus, the same pattern will repeat starting from $x_1^{i_2}$. Thus, the pattern of exploration by the $k$ agents shifts by the vector $v = x_1^{i_2} - x_1^{i_1}$.
Let $r_1$ denote the number of grid points visited by the agents until time $i_1$, let $r_2$ denote the number of grid points visited by the agents between times $i_1$ and $i_2$. Finally, let $r = \max(r_1, r_2)$. Therefore all the grid points visited by the agents from $t$ onward fall within the cylinder in direction $v$ of radius $r$ and origin $x_1^t$. It is clear that $v$ and $r$ depend only on $rconf_t$.
\end{proof}

The following lemma collects several facts about the behaviors of 1, 2, and 3 agents, respectively.

\begin{lemma}\label{lem:useful-facts}
\begin{enumerate}
\item Consider moves of a single agent $a$ in between meetings with other agents. The agent $a$ explores lattice points that fall within a cylinder in direction $v$ of radius $r$. The direction $v$ and the radius $r$ depend only on the state of $a$ at the beginning of the movement.
\item Consider moves of two agents $a_1$ and $a_2$ that start from the same grid point until one of the agents meets an agent different from $a_1$ or $a_2$. Then only one of the following is possible:
\begin{enumerate}
\item Agent $a_i$ visits grid points that fall within a cylinder in direction $v_i$ of radius $r_i$, where $v_i$ and $r_i$ depend on the pair of states of agents $a_1$ and $a_2$ at the beginning of the movement; or
\item Both agents visit grid points that fall within a single cylinder in direction $v$ of radius $r$, where $v$ and $r$ depend on the pair of state of agents $a_1$ and $a_2$ at the beginning of the movement.
\end{enumerate}
\item Consider 3 agents in $\mathbb{R}^3$ that run a protocol for exploring all of $\mathbb{Z}^3$. All three agents cannot meet simultaneously infinitely often.
\end{enumerate}
\end{lemma}
\begin{proof}
All the statements are easy consequences of Lemma~\ref{lem:rel-conf-helper}.
\begin{enumerate}
\item Consider $a$ that never meets any other agent. Since $a$ can have finitely many states and its relative configuration is $(0, q^i)$, where $q^i$ is the state of $a$ at time $i$, some relative configuration has to repeat. Then by Lemma~\ref{lem:rel-conf-helper} all grid points visited by $a$ lie within some cylinder with direction $v$ and radius $r$ that depend on the state of $a$ at the beginning of the considered time interval.

We now claim that the cylinder defined as above covers all grid points that are visited by $a$ even if the number of steps until meeting another agent is finite. That is because the path of agent $a$ until time $t_1$ is a subpath of the path of agent $a$ until time $t_2 > t_1$, if $a$ does not meet any other agent until $t_1$.

\item First suppose that $a_1$ and $a_2$ never meet an agent different from $a_1$ and $a_2$. Then there are two possibilities: either (a) agents $a_1$ and $a_2$ meet each other finitely many times, or (b) agents $a_1$ and $a_2$ meet each other infinitely often. 

In case (a), consider the step immediately after the last time $a_1$ and $a_2$ meet each other. By assumption, $a_1$ and $a_2$ do not meet any other agents, thus we can apply the first part of this lemma to each of them separately. By adjusting the radius of the cylinders we can also cover all grid points that were visited until the last time $a_1$ and $a_2$ met each other. This proves the first subpart.

In case (b), consider relative configurations at times when $a_1$ and $a_2$ meet each other. These relative configurations are of the form $(0, 0, q_1^i, q_2^i)$. Thus, there are only finitely many possible relative configurations. Since $a_1$ and $a_2$ meet each other infinitely often, they must repeat some relative configuration. The second subpart of the statement follows by Lemma~\ref{lem:rel-conf-helper}.

If one of the two agents $a_1$ and $a_2$ eventually meet an agent different from $a_1$ and $a_2$ then by the same argument as in the proof of the first part of the lemma the grid points visited until then still fall within the cylinders defined above.

\item Consider relative configurations at times of the meetings. They are of the form \\
$(0, 0, 0, q_1^i, q_2^i, q_3^i)$. Therefore, there are only finitely many possible relative configurations. Assume for contradiction that the agents meet infinitely often, then some relative configuration has to repeat. By Lemma~\ref{lem:rel-conf-helper} all the grid points visited by the three agents fall within a cylinder. Since all of $\mathbb{Z}^3$ cannot be contained within a single cylinder, we get a contradiction.
\end{enumerate}
\end{proof}

\begin{lemma}\label{lem:meet-io}
Consider 3 agents in $\mathbb{R}^3$ that run a protocol for exploring   $\mathbb{Z}^3$. Then each agent has to meet with some other agent infinitely often.
\end{lemma}
\begin{proof}
Suppose for contradiction that we have an agent that meets the two other agents only finitely often. Then by Lemma~\ref{lem:useful-facts}, it explores only grid points that fall within some cylinder $C_1$. By the same lemma, the two remaining agents either explore their own cylinders $C_2$ and $C_3$, or they explore a combined single cylinder $C_4$. Thus, we get that $\mathbb{Z}^3$ can be covered by either 2 or 3 cylinders, which contradicts Lemma~\ref{lem:finitely-many-cylinders}. 
\end{proof}

The following lemma justifies why the number of meeting points between two agents where the third is at some distance $d$ is bounded by a constant independent of $d$.

\begin{lemma}\label{lem:finite-coordinates}
Consider 3 agents in $\mathbb{R}^3$ that run a protocol for exploring  $\mathbb{Z}^3$. By part (3) of Lemma~\ref{lem:useful-facts}, we can consider $t$ large enough such that all three agents never meet after time $t$. Consider only those times after $t$ such that two of the three agents are collocated: $t_1, t_2, \ldots$. Define $d(t_i)$ to be the distance between the two collocated agents at time $t_i$ and the lone agent. There is a universal constant $k$ such that for all $d$ we have
\[ |\{ i \mid d(t_i) = d\}| \le k.\]
\end{lemma}
\begin{proof}
The essence of the proof is to show that only finitely (i.e., at most $k$) different relative configurations can give rise to the same value $d$. If we prove this, then it means that in the entire trace of the exploration, the value $d$ can be incurred at most $k$ times. Otherwise,  some relative configuration would have to repeat, meaning by Lemma~\ref{lem:rel-conf-helper} that the three agents only explore grid points falling within a cylinder, so they cannot explore the entire $\mathbb{Z}^3$.

Now, consider the situation where $a_1$ and $a_2$ are collocated and $a_3$ is at distance $d(t_1)$ at time $t_1$. Moreover, suppose that the next time when two agents meet each other will be when $a_2$ meets with $a_3$ at some later time. We need to bound the number of possible coordinates of $a_3$ relative to $a_1$ at the beginning of the movement, i.e., at $t_1$. There are only finitely many directions and radii of cylinders within which $a_2$ can move, and similarly for $a_3$ (since they depend only on the states of agents). Consider one such cylinder for $a_2$ with direction $v_2$ and radius $r_2$. Similarly, consider one such cylinder which corresponds to $a_3$ with direction $v_3$ and radius $r_3$. The possible starting locations\footnote{Relative to $a_1$.} $y_3$ of $a_3$ have to satisfy
\begin{enumerate}
\item $||y_3||_1 = d(t_1)$,
\item $||v_2(t-t_1) - (v_3(t-t_1)+y_3)||_\infty \le r_1 + r_2$.
\end{enumerate}
It is easy to see that there are only finitely many vectors $y_3 \in \mathbb{Z}^3$ that satisfy the above conditions. First note that there are only finitely many error terms $e \in \mathbb{Z}^3$ such that $||e||_\infty \le r_1 + r_2$. Thus, it is sufficient to fix one such $e$ and show that there are finitely many $y_3$ that satisfy $||y_3||_1=d(t_1)$ and $v_2(t-t_1)-(v_3(t-t_1)+y_3)= e$. We can rewrite the second condition as $y_3 = v_2 (t-t_1) - v_3 (t-t_1) - e$, which gives us three equations and 4 unknowns (coordinates of $y_3$ and $t$). The fourth equation is given by $||y_3||_1=d(t_1)$. We can consider this to split into 8 cases depending on signs of coordinates of $y_3$, each case giving at most one solution to the overall system. Overall, we get that for each $v_2, r_2$ and each $v_3, r_3$ there can only be finitely many starting points $y_3$ of $a_3$ such that $a_2$ and $a_3$ do not miss each other, while exploring $\mathbb{Z}^3$ along the corresponding cylinders. Since the number of possible values of $v_2, r_2, v_3, r_3$ is also finite, the statement of the lemma follows.
\end{proof}

Combining the results proven so far, we can show the key lemma.

\begin{theorem}
Consider 3 agents in $\mathbb{R}^3$ that run a protocol for exploring  $\mathbb{Z}^3$. Suppose that by time $t$ the maximum distance of an agent from the true origin $0$ is at most $d$. Then the number of grid points visited by all three agents by time $t$ is $O(d^2)$. 
\end{theorem}
\begin{proof}
The assumption implies, in particular, that at all meeting times $t_i \le t$ the distances satisfy $d(t_i) \le d$.  Between meeting times, the agents explore along some cylinders. Whenever an agent explores along a cylinder, the agent visits only a linear number of grid points (since width of the cylinder is constant). Thus, an agent can only explore $O(d)$ grid points between two meetings. By Lemma~\ref{lem:finite-coordinates} we have $|\{i \mid d(t_i) \le d\}|\le d k$, i.e., there can only be a linear (in $d$) number of meetings. Multiplying the two estimates gives an upper bound $O(d^2)$ on the total number of visited grid points.
\end{proof}

The following easy corollary is one of the main conclusions of this subsection.

\begin{corollary}
Consider 3 agents in $\mathbb{R}^3$ that run a protocol for exploring  $\mathbb{Z}^3$. In order to visit all grid points in the ball of radius $D$ the distance of some agent from the origin must have been $\Omega\left(D^{3/2}\right)$.
\end{corollary}

We believe that it should be possible to extend the above result to the case of general $n$ and $k$ agents. Namely, the desired statement is that if $k$ agents visit all grid points inside the ball of radius $D$ then one of the agents has to visit a location that is $\Omega\left(D^{\frac{n}{k-1}}\right)$ away from the origin. The proof, which would be based on extended versions of Lemmas~\ref{lem:meet-io} and~\ref{lem:finite-coordinates}, would proceed by induction on the number of agents. The formal proof is deferred to the complete version of the paper.

\section{Unoriented Grids}
\label{sec:unoriented}
In this section we consider exploration of unoriented grids. In such grids, the incident edges to each node are labelled by different labels from $\{0,1,\dots, 2n-1\}$. Note that each edge receives two labels (port numbers), one on each end. However, no global consistency among edge labels can be assumed. This, together with agent's finite memory, means that a lone agent cannot cross any non-constant distance, as the irregular nature of the port labels would lead it astray, never to meet any other agent. It is, therefore, an interesting question to ask ``How many additional agents are necessary to solve the problem in unoriented grids?''

\begin{figure}[h]
  \centerline{\includegraphics[width=12cm]{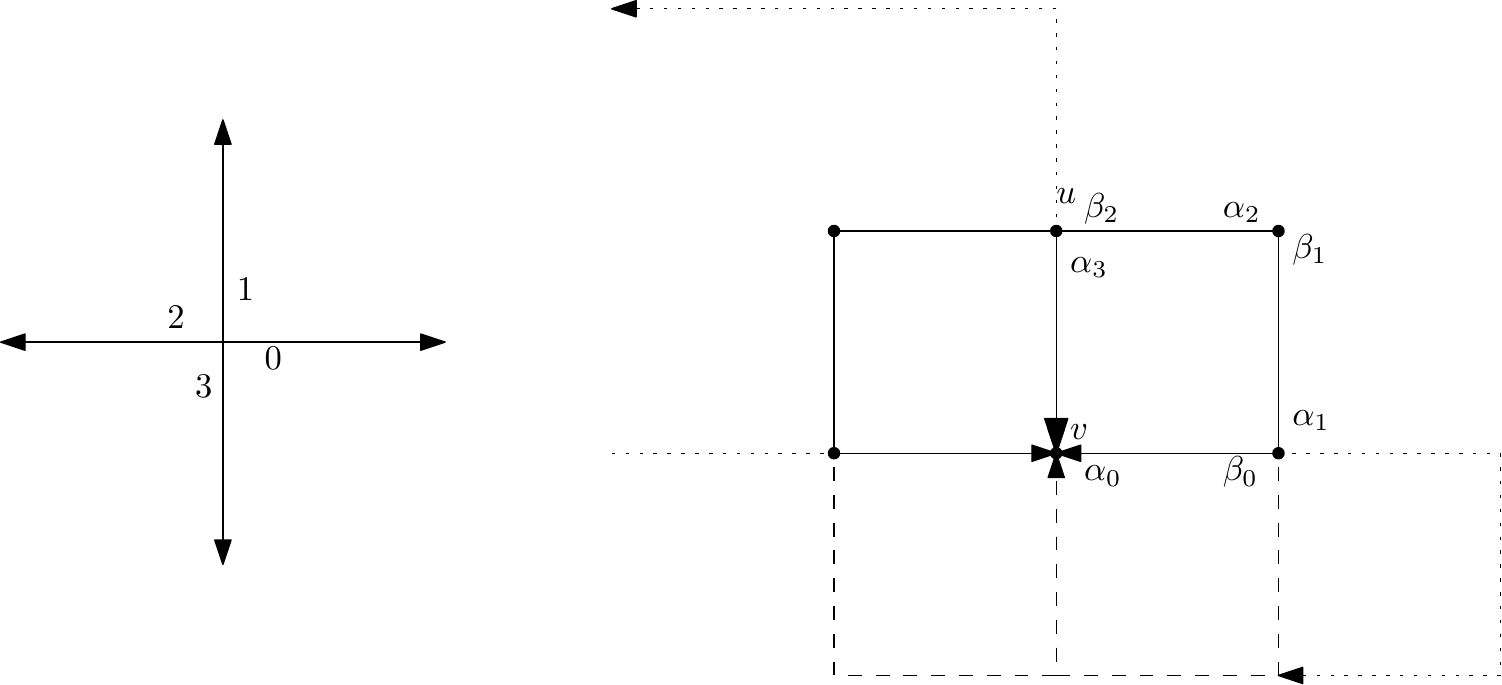}}
  \caption{Left: The global directions. Right: Handrail in 2D. Full lines correspond to $\alpha$ that lead to computing $v_0$ and $v_2$. Dashed lines returned to $v$ but did not satisfy $\alpha_3 = 3$. Dotted lines did not return to $v$ or backed on themselves}\label{fig:unor}
\end{figure}

In this section we show that one additional agent is sufficient in the semi-synchronous model, while two additional agents are sufficient in the synchronous model. This result is obtained by employing and generalizing from $2D$ the {\em handrail} technique of \cite{Mans97}, which allows a moving agent to establish by local exploration the relationship between the port labels of vertices it is moving through, in effect carrying the orientation along. 
However, an auxiliary agent is needed to achieve this.  As all our algorithms using $k$ semi-synchronous agents are in fact algorithms for one active agent using the remaining agents as tokens with IDs, a single auxiliary agent is sufficient. In the case of synchronous agents, at most two of them are active agents, which implies that two auxiliary agents are sufficient.

In the remainder of this section, we sketch the {\em handrail} technique. Combining it with the algorithms for the oriented grid is straightforward for the semi-synchronous case, while a bit more care about the timing is needed in the synchronous case.

Let $\oplus$ denote addition modulo $2n$. Let $i\in 0,1,\dots, n-1$. Then the global direction $i$ corresponds to increasing the position in dimension $i$, while direction $i+n$ corresponds to decreasing the position in dimension $i$.

Let $i_u$ for $i\in\{0,1,\dots, 2n-1\}$ denote the port label at node $u$ of the edge leaving $u$ in the direction $i$, and let $O_u$, the orientation at $u$, denote $\{i_u\}_{i=0}^{2n-1}$.

Assume the agent $a$ knows the orientation $O_u$ at the current node $u$ and it wants to move, as in the algorithm for oriented grid, to $u$'s neighbour $v$.
Applying the following procedure allows $a$ to compute $O_v$; this means $a$ can maintain the global orientation while moving.
\begin{algorithm}[htb]
	\caption{Handrail: Maintaining orientation when moving from $u$ to $v$}
	\label{alg:handrail}
	\begin{algorithmic}[1]
		\State On entry: $a$ is in $u$ and knows $O_u$, wants to move to $u$'s neighbour $v$.
		\State On exit: $a$ is in $v$ and knows $O_v$.
		\State $a$ goes to $v$ along direction $i$
        \State let $p$ be the label of the arrival port in $v$. Set ${i\oplus 2}_v = p$
        \State $a$ drops the auxiliary token/agent in $v$
		\For {all $\alpha\in \{0,1,\dots, 2n-1\}^4$}
			\State $a$ interprets $\alpha$ as a sequence of port labels and executes the corresponding $4$-step walk, collecting in $\beta$ the arrival ports along this walk
            \If {$a$ is not collocated with its auxiliary token, or $\alpha_{j+1}=\beta_j$ for some $j\in \{0,1,2\}$}
                \State return to $v$ using $\beta$ and proceed to the next $\alpha$
            \ElsIf {$\alpha_3 = i$ and $\beta_2 = j_u$}
                \State set $j_v = \alpha_0$
            \Else
                \State proceed to the next $\alpha$
            \EndIf
		\EndFor
        \State set $i_v$ to the remaining unused port label
	\end{algorithmic}
\end{algorithm}

The correctness of Algorithm \ref{alg:handrail} is based on the fact that in a grid, the only way to return to $v$ via direction $i$ after a $4$-step walk which never backtracks on itself is when the four steps were in directions $j,-i,-j,i$ for $j\notin \{i,-i\}$.
Note that the cost, i.e., the number of moves or time, of  Algorithm~\ref{alg:handrail} is $O(n^4)$, i.e. a constant. With a more careful approach (at a cost of more complex presentation), this can be reduced to $O(n^3)$.

\section{Conclusions and Open Questions}
\label{sec:conclusion}
We studied the exploration of $n$-dimensional grids for $n \geq 3$ by finite state automata agents. We showed the surprising result that three randomized synchronous agents suffice to find a treasure in an $n$-dimensional grid for any $n$; this is optimal in the number of agents. Our strategy can also be implemented by four randomized asynchronous agents, or four deterministic synchronous agents, or five deterministic asynchronous agents. For the three-dimensional case, we gave a different algorithm for the deterministic asynchronous case that uses only 4 agents, and is optimal. Our algorithms for $n\ge 4$ require agents to travel far away from the origin, i.e., exponential in $D$ distance away, while looking for a treasure which is located at distance $D$ from the origin. We also considered the question of whether it is possible to design algorithms that use few agents and do not require travelling much further than distance $D$ away from the origin in order to explore the entire ball of radius $D$ around the origin. We answered the question positively by describing an algorithm that uses $O(\sqrt{n})$ semi-synchronous deterministic agents that never travel beyond $2D$ while exploring the ball of radius $D$. We also showed that $3$ synchronous deterministic agents in $3$ dimensions must travel $\Omega(D^{3/2})$ away from the origin.

There remain many interesting open questions on the exploration of the $n$-dimensional grids. Is it possible for 4 deterministic semi-synchronous agents to explore an $n$-dimenstional grid for $n \geq 4$? For $n \geq 3$, can  exploration of an $n$-dimensional grid be achieved 3 randomized semi-synchronous agent or deterministic synchronous agents? What is the minimal number of agents that achieve polynomial time exploration? What is the minimal number of agents such that the distance of the furthest visited node from the origin is limited to polynomial in $D$ (eg. linear in $D$; $D+o(D)$)? Is it possible to save an agent in the case of synchronous unoriented grids?

\bibliographystyle{plain}
\bibliography{ref} 
\end{document}